\documentclass[conference]{IEEEtran}
\IEEEoverridecommandlockouts
\usepackage{cite}
\usepackage{amsmath,amssymb,amsfonts}
\usepackage{algorithm}
\usepackage{algpseudocode}
\usepackage{graphicx}
\usepackage{textcomp}
\usepackage{bbm}
\usepackage{xcolor}
\def\BibTeX{{\rm B\kern-.05em{\sc i\kern-.025em b}\kern-.08em
    T\kern-.1667em\lower.7ex\hbox{E}\kern-.125emX}}

\usepackage{tikz}
\usepackage{pgfplots}
\usetikzlibrary{spy}
\usepackage[normalem]{ulem}

\def\*#1{\mathbf{#1}}
\newcommand{\vect}[1]{{\mathbf{#1}}}
\newcommand{\mat}[1]{{\mathbf{#1}}}

\newcommand{\rsubsK}{\left[ [\sf K] \right]^{\sf r}}

\renewcommand{\sf}[1]{\mathsf{#1}}

\newcommand{\cT}{\mathcal{T}}
\newcommand{\cL}{\mathcal{L}}
\newcommand{\cB}{\mathcal{B}}

\newcommand{\K}{\mathsf{K}}
\renewcommand{\r}{\mathsf{r}}

\newcommand{\Ufi}{$\mat{U}_5$}

\newcommand{\lo}[1]{{\color{magenta}#1}}
\renewcommand{\lo}[1]{{#1}}
\newcommand{\sh}[1]{{\color{green}#1}}
\renewcommand{\sh}[1]{}
\newif\ifADDpagenumber
\ADDpagenumberfalse

\allowdisplaybreaks[4]
\usepackage{amsthm}
\newtheorem{theorem}{Theorem}
\newtheorem{lemma}{Lemma}
\newtheorem{corollary}{Corollary}

\begin{document}

\title{Necessity of Cooperative Transmissions for Wireless MapReduce}

\sloppy
\allowdisplaybreaks[4]

\author{
    \IEEEauthorblockN{Yue Bi\IEEEauthorrefmark{1}, Mich\`ele Wigger\IEEEauthorrefmark{1}}
    \IEEEauthorblockA{\IEEEauthorrefmark{1} \textit{LTCI, Telecom Paris, IP Paris}, 91120 Palaiseau, France; 
    \{bi, michele.wigger\}@telecom-paris.fr}
}

\maketitle

\ifADDpagenumber
\thispagestyle{plain} 
\pagestyle{plain}
\fi

\begin{abstract}
The paper presents an improved upper bound (achievability result) on the optimal tradeoff between Normalized Delivery Time (NDT) and computation load for distributed computing MapReduce systems in certain ranges of the parameters. The upper bound is based on interference alignment combined with zero-forcing. The paper further provides a lower bound (converse) on the optimal NDT-computation tradeoff that can be achieved when IVAs are partitioned into sub-IVAs, and these sub-IVAs are then transmitted (in an arbitrary form) by a single node, without cooperation among nodes. For appropriate linear functions (e.g., XORs), such non-cooperative schemes can achieve some of the best NDT-computation tradeoff points so far obtained in the literature. However, as our lower bound shows, any non-cooperative scheme achieves a worse NDT-computation tradeoff than our new proposed scheme for certain parameters, thus proving the necessity of cooperative schemes like zero-forcing to attain the optimal NDT-computation tradeoff.
\end{abstract}

\section{Introduction}

MapReduce systems are becoming important building blocks also in wireless scenarios because they allow us to distribute heavy computation tasks over multiple distributed nodes. Important performance parameters of such systems are the  \emph{computation load} $\sf r$, which characterizes the average amount of data stored at the nodes (i.e., memory requirements), and the  \emph{Normalized Delivery Time (NDT)} $\Delta$, which characterizes the duration of the wireless shuffle phase (i.e., communication requirements).  The optimal tradeoff between memory requirements and communication duration was first characterized in \cite{li_fundamental_2018}, but under the assumption of perfect multicast communication links. The optimal scheme for this noiseless setup splits intermediate values (IVA) into message parts and then communicates appropriate XORs of these parts. Decoding requires that receiving nodes use their knowledge of the message parts transmitted by other nodes to cancel non-desired parts of the XORs. 
Similar assumptions of ideal multicast or noiseless communication links have also been adopted in several related works, e.g.,~\cite{yan2019storage, xu_new_2021, wan_distributed_2022}.

Several subsequent works \cite{li_wireless_2019,bi_dof_2022,bi_normalized_2024,yuan_coded_2022} have studied the \emph{NDT-computation tradeoff} of \emph{wireless} MapReduce, i.e., the smallest NDT that can be achieved under a constraint on the computation load. The problem was first investigated for wireless interference networks in \cite{li_wireless_2019}, which proposed a coding scheme based on one-shot beamforming and zero-forcing (ZF).  In our previous works in \cite{bi_dof_2022,bi_normalized_2024}, we  obtained improved NDT-computation tradeoffs by introducing interference alignment (IA) \cite{cadambe_interference_2008, jafar_degrees_2008} to the shuffle phase. For  computation loads $\sf r=\lfloor (\sf K-1 )/2\rfloor$, where $\sf K$ denotes the number of nodes and is assumed to be odd, our work in \cite{bi_normalized_2024}  also  proposes  an improved IA scheme based on ZF. Notice that ZF introduces complex dependencies between IA matrices, rendering  performance analysis of  IA schemes involved \cite{annapureddy_degrees_2012}. 

In this paper, we improve our IA \& ZF scheme in \cite{bi_normalized_2024}. Our new scheme applies to an arbitrary number of nodes  $\sf K$ but still requires that $\sf r=\lfloor( \sf K-1 )/2\rfloor$. For odd values of $\sf K$ it strictly improves over our previous scheme in \cite{bi_normalized_2024}. The new scheme is again based on IA \& ZF and the analysis requires proving algebraic independence of certain functions. We provide an analytic proof for $\sf K=5$ in this paper as well as numeric verification for $\sf K\in\{6,\ldots, 15\}$.

Inspired by the excellent performance of the scheme in \cite{li_fundamental_2018} and to avoid complicated analyses, several schemes (also for related cache-aided scenarios) \cite{hachem_degrees_2018, bi_dof_2022,bi_normalized_2024} focused on exploiting the multicast behaviour of wireless MapReduce communication and the IVA side-information at the nodes, but without any form of cooperation between transmit signals. Specifically, the schemes divide the IVAs into submessages and each submessage is transmitted (in an arbitrary form) by only one of the nodes, without cooperation from the other nodes. Transmit signals across nodes are thus  independent of each other. Nodes still use all their side-information about IVAs in their decoding steps. In this work, we provide a lower bound (converse result) on the NDT-computation tradeoff that can be achieved by any such non-cooperative scheme. We further show that this lower bound lies above the  NDT-computation tradeoffs obtained by our new IA \& ZF scheme, thus proving strict suboptimality of non-cooperative schemes. Notice that the previously obtained NDT-computation tradeoffs do not improve over the non-cooperative lower bound and do not allow us to obtain the desired conclusion.

\textit{Notations:} We use standard notation, and also
define $[n] \triangleq \{1,2,\ldots,  n\}$ and  $[\mathcal{A}]^t$ as the collection of all the subsets of $\mathcal{A}$ with cardinality $t$. 
i.e. $[\mathcal{A}]^t \triangleq \{\cT\colon \cT \subseteq \mathcal{A}, |\cT|=t\}$. 
%

\section{Wireless MapReduce Framework}\label{sec:WDC_system}
Consider a wireless MapReduce system with $\sf K \geq 5$ nodes labeled $1,\ldots,\sf K$, with $\sf M$ input files $W_1,\ldots,W_{\sf M}$, and with $\sf K$ output functions $\phi_1,\ldots,\phi_{\sf K}$, where function $\phi_k$ is assigned to Node~$k$.
The system follows the MapReduce framework, in which each output function is decomposed as
\begin{equation}
	\phi_q(W_1,\ldots,W_{\sf M}) = v_q(a_{q,1},\ldots,a_{q,\sf M}), \quad q\in[\sf K],
\end{equation}
where
\begin{equation}
	a_{q,p} = u_{q,p}(W_p), \quad p\in[\sf M],
\end{equation}
denotes the intermediate value (IVA) computed from file $W_p$.
All IVAs are assumed independent and consist of $\sf A$ i.i.d. bits.

The MapReduce framework has 3 phases:
\paragraph{Map phase}
Each node $k$ is assigned a subset of files $\mathcal{M}_k\subseteq[\sf M]$ and computes the IVAs
$\{a_{q,p} : p\in\mathcal{M}_k,\ q\in[\sf K]\}$.

\paragraph{Shuffle phase}
Nodes exchange IVAs over $\sf n$ uses of a wireless channel in full-duplex mode.
Node~$k$ transmits
\begin{equation}\label{eq:comp_encoding}
	\vect{X}_k
	=
	f_k^{(\sf n)}\!\left(\{a_{1,p},\ldots,a_{\sf K,p}\}_{p\in\mathcal{M}_k}\right),
\end{equation}
subject to the average power constraint
\begin{equation}\label{eq:power}
	\frac{1}{\sf n}\sum_{t=1}^{\sf n}\mathbb{E}[|X_k(t)|^2] \le \sf P.
\end{equation}
Node~$k$ observes
\begin{equation}\label{eq:channel}
	Y_k(t)
	=
	\sum_{k'\in[\sf K]\backslash\{k\}} H_{k,k'}(t)X_{k'}(t)
	+ Z_k(t),
	\quad t\in[\sf n],
\end{equation}
and $\vect{Y}_k \triangleq \left(Y_k(1) \dots, Y_k(\sf n)\right)$, where $\{H_{k,k'}(t)\}$ and $\{Z_k(t)\}$ are i.i.d. over time and across node pairs. Each coefficient $H_{p,q}(t)$ has independent and identically distributed real and imaginary parts, following a specified continuous distribution over a bounded interval $[-\sf H_{\max}, \sf H_{\max}]$ for some positive integer $\sf H_{\max}$.
All nodes have perfect channel state information.

Based on $\vect{Y}_k$ and the IVAs computed locally, Node~$k$ decodes the missing IVAs as
\begin{equation}\label{eq:comp_decoding}
	\hat{a}_{k,p}
	=
	g_{k,p}^{(\sf n)}\!\left(
	\{a_{q,i}\}_{i\in\mathcal{M}_k,q\in[\sf K]}, \vect{Y}_k
	\right),
	\quad p\notin\mathcal{M}_k .
\end{equation}

\paragraph{Reduce phase}
Each node $k\in[\sf K]$ applies its reduce function $v_k$ to its available IVAs  to compute its assigned output function $\phi_k$.

The performance of the distributed computing system is measured by two factors: \textit{computation load} and \textit{normalized delivery time (NDT)}.
The computation load is defined as
\begin{equation}\label{eq:r}
	\sf r \triangleq \sum_{k\in[\sf K]} \frac{|\mathcal{M}_k|}{\sf M},
\end{equation}
and NDT is defined as
\begin{equation}\label{eq:Delta}
	\sf \Delta
	\triangleq
	\varliminf_{\sf P\to\infty}
	\varliminf_{\sf A\to\infty}
	\frac{\sf n}{\sf A\,\sf K\,\sf M}\log \sf P .
\end{equation}

We define the optimal NDT--computation tradeoff $\Delta^*(\sf r)$ as the infimum of all $\sf \Delta$
for which there exist file assignments $\{\mathcal{M}_k\}$ and encoding/decoding functions
satisfying
\begin{equation}\label{eq:error_computing}
	\Pr\!\left[
	\bigcup_{k\in[\sf K]}
	\bigcup_{p\notin\mathcal{M}_k}
	\hat{a}_{k,p}\neq a_{k,p}
	\right]
	\to 0
	\quad \text{as } \sf A\to\infty .
\end{equation}

In this work, we choose a the well known combinatorial file assignment for the wireless MapReduce system.
Specifically, the input files $\{W_1,\ldots,W_{\sf M}\}$ are partitioned into $\binom{\sf K}{\sf r}$ disjoint bundles, each associated with a size-$\sf r$ subset $\mathcal{T}\in [[\sf K]]^{\sf r}$ and stored at the nodes in $\mathcal{T}$.
Since each file is stored at exactly $\sf r$ nodes, this assignment satisfies the computation load constraint.

\subsection{Relation with $\sf r$-fold cooperative channel}
Under this combinatorial file assignment, the Shuffle phase naturally induces a wireless channel with $\sf r$-fold transmitter cooperation.
In particular, consider the wireless interference channel described in \eqref{eq:channel} under the average power constraint \eqref{eq:power}.
For a given $\sf r\in[\sf K]$, any subset of $\sf r$ nodes $\mathcal{T}\in [[\sf K]]^{\sf r}$ jointly holds a message $a_{k, \mathcal{T}}$ intended for Node~$k$, for each $k\in[\sf K]\setminus\mathcal{T}$.
Each message $a_{k, \mathcal{T}}$ is uniformly distributed over $\{1,\ldots,2^{\sf n R_{k, \mathcal{T}}}\}$.

A rate tuple $(R_{k, \mathcal{T}} : \mathcal{T}\in [[\sf K]]^{\sf r},\ k\in[\sf K]\setminus\mathcal{T})$ is said to be achievable if there exist encoding and decoding functions such that the decoding error probabilities vanish as $\sf n\to\infty$.
The corresponding sum degrees of freedom is defined as
\begin{equation}
	\mathrm{SDoF}(\sf r) \triangleq \lim_{P\to\infty} \sup_{\vect{R} \in \mathcal{C}(P)}
	\sum_{\cT\in \rsubsK} \sum_{k\in[K]\setminus\cT} \frac{R_{k,\cT}}{\log P}.
\end{equation}
where the supremum is over all achievable rate tuples under power constraint $\sf P$.

We have the following lemma, which establishes a direct connection between the NDT and the SDoF \cite{bi_normalized_2024}.

\begin{lemma}\label{lem:E}Under the combinatorial file assignment, 
	for any $\sf r\in[\sf K]$,
	\begin{equation}\label{eq:DSDof}
		\sf \Delta(\sf r)
		=
		\left(1-\frac{\sf r}{\sf K}\right)\frac{1}{\mathrm{SDoF}(\sf r)}.
	\end{equation}
\end{lemma}

\section{Main results}

Our first result is an  achievable NDT-computation tradeoff  that improves over previous results for $\sf r = \lfloor( \sf K - 1 )/2\rfloor$.
\begin{theorem}[Achievable NDT-computation tradeoff]
	\label{thm:achievable}
	The NDT-computation tradeoff for $5 \leq \sf K \leq 15$\footnote{
		The restriction $\sf K \leq 15$ is due to the use of numerical verification.
		This choice does not suggest any inherent limitation of the proposed scheme for larger $\sf K$, instead, 
		The performance of the scheme has only been numerically validated for $\sf K \leq 15$. }
	and $\sf r = \lfloor (\sf K - 1)/2 \rfloor$ satisfies
	\begin{equation}\label{eqn:achievable}
		\sf \Delta(\sf r) \leq
		\begin{cases} 
			\dfrac{1}{\sf K}\left(1 - \dfrac{\sf r}{\sf K}\right)\left(1 + \dfrac{4}{(\sf K -2)^2}\right), & \sf K \text{ even}, \\[0.5ex]
			\dfrac{1}{\sf K}\left(1 - \dfrac{\sf r}{\sf K}\right)\left( 1 + \dfrac{1}{(\sf K - 1)(\sf K - 2)}  \right), & \sf K \text{ odd}.
		\end{cases}
	\end{equation}
\end{theorem}
\begin{IEEEproof}
	The coding scheme and corresponding SDoF for the $\sf r$-fold cooperative channel is described in Section~\ref{sec:scheme}. Above tradeoff can then be obtained by applying Lemma \ref{lem:E}.
\end{IEEEproof}
As illustrated in Fig.~\ref{fig:compare}, for $\sf r = \lfloor (\sf K -1)/2 \rfloor$, the proposed scheme strictly outperforms the schemes in \cite{li_wireless_2019, bi_dof_2022, bi_normalized_2024} for all considered values of $\sf K$.
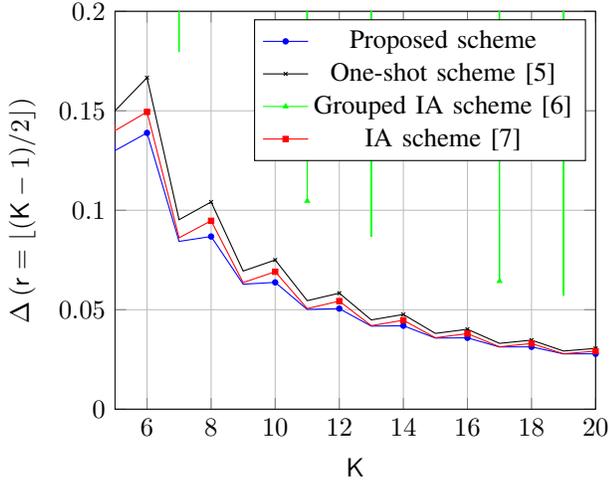
\begin{figure}[t]
	\centering
	\begin{tikzpicture}
		\begin{axis}[ width=0.9\linewidth,
			grid=both,
			xlabel={$\sf K$},
			xmin=5, xmax=20,
			ylabel={$\Delta \left( \sf r =\lfloor (\sf K-1)/2 \rfloor \right)$},
			scaled y ticks=false,
			yticklabel style={/pgf/number format/fixed},
			ymin=0, ymax=0.2,
			]
			\addplot[color=blue, mark=*, mark size=1pt, mark repeat=2, mark phase=2]table [x index=0, y index=1] {figures/delta_single.dat};
			\addlegendentry{Proposed scheme}
			
			\addplot[color=black, mark=x, mark size=1pt, mark repeat=2, mark phase=2]table [x index=0, y index=1] {figures/deltaub_oneshot.dat};
			\addlegendentry{One-shot scheme \cite{li_wireless_2019}}
			
			\addplot[color=green, mark=triangle*, mark size=1pt, mark repeat=2, mark phase=2]table [x index=0, y index=1] {figures/deltaub_grouped.dat};
			\addlegendentry{Grouped IA scheme \cite{bi_dof_2022}}
			
			\addplot[color=red, mark=square*, mark size=1pt, mark repeat=2, mark phase=2]table [x index=0, y index=1] {figures/deltaub_tit.dat};
			\addlegendentry{IA scheme \cite{bi_normalized_2024}}
		\end{axis}
	\end{tikzpicture}
	\caption{Comparison of the achieved NDT--computation tradeoff between the proposed scheme and the schemes in~\cite{li_wireless_2019, bi_dof_2022, bi_normalized_2024}.}
	\label{fig:compare}
\end{figure}

Above achievability result is obtained using a coding scheme based on IA and ZF. A key-characteristics of ZF is that nodes do not send independent signals, but correlate them so as to avoid that certain signals interfere at some of the nodes. Our second result implies that under the proposed combinatorial file assignment, such cooperative transmission is essential to achieve the best NDT-computation tradeoff. In  the following Theorem~\ref{thm:split_converse}, we provide a lower bound on the best NDT-computation tradeoff that is achievable using any kind of non-cooperative transmission. As we show, this lower bound lies above the NDT-computation tradeoff achieved by our IA and ZF scheme in Theorem~\ref{thm:achievable}.

Specifically, consider any MapReduce system with combinatorial file assignment and where each IVA $a_{q,p}$ is split into $\sf r$ equal length sub-messages $(a^q_{t, \mathcal{T}})_{t\in\mathcal{T}}$ where $\mathcal{T}$ is the set of nodes that have stored $a_{q,p}$, i.e., $\mathcal{T}=\{k\colon p \in \mathcal{M}_k\}$. Each node $t \in \cT$ then transmits only $a^q_{t,\mathcal{T}}$ but not the other submessages implying node $t$'s transmit signal $X_{t}^n$ is independent of the sub-messages $(a^q_{t', \mathcal{T}})_{t' \in\mathcal{T}\backslash\{t\}}$. In contrast, node $t$ is allowed to use its knowledge of all messages $(a^q_{t', \mathcal{T}})_{t'\in\mathcal{T}}$ during its decoding.  Notice that this set of strategies in particular includes  XOR-based inter-message coding as proposed in \cite{hachem_degrees_2018, li_fundamental_2018} where each XOR is independently multi-cast to a set of receivers that can remove the non-desired parts in the XOR. 

\begin{theorem}[Converse under  Non-Cooperative Transmissions]
	\label{thm:split_converse}
Under combinatorial file assignments and  non-cooperative transmissions of sub-IVAs, the achievable NDT-computation tradeoff is bounded as:
	\begin{equation}\label{eqn:split_converse}
		\sf \Delta(\sf r) \geq \dfrac{1}{\sf K}\left(1-\dfrac{\sf r}{\sf K}\right) \cdot \dfrac{(\sf K - 2)\sf r + \sf K - 1}{\sf r (\sf K - 1)}.
	\end{equation}
\end{theorem}
\begin{IEEEproof}
	The converse proof for the $r$-fold cooperative channel with message-splitting approach is given in Section~\ref{sec:converse}. The tradeoff can be obtain by applying Lemma~\ref{lem:E}.
\end{IEEEproof}

Then, we obtain the following corollary.
\begin{corollary}
	For $\sf K \geq 5$ and $\sf r = \lfloor (\sf K - 1)/2 \rfloor$  the bound in \eqref{eqn:achievable} is strictly smaller than the bound in \eqref{eqn:split_converse}.
\end{corollary}
\begin{IEEEproof}
	For ${\sf K}$ even and $\sf r =  \lfloor (\sf K - 1)/2 \rfloor=(\sf K -2)/2$, the expression in 
	\eqref{eqn:split_converse} evaluates to:
	\begin{equation} \dfrac{1}{\sf K} \left(1 - \dfrac{\sf r}{\sf K}\right)\cdot\dfrac{(\sf K-2)^2 + 2(\sf K-1)}{(\sf K-2)(\sf K-1)}. 
	\end{equation}
	which is strictly larger than the achievable bound in \eqref{eqn:achievable} for all even $\sf K \geq 6$.
	Similarly, for  $\sf K$ odd and  $\sf r =  \lfloor (\sf K - 1)/2 \rfloor= (\sf K-1)/2$, the bound in \eqref{eqn:split_converse} simplifies to 
	\begin{equation}\dfrac{1}{\sf K (\sf K -1)}\left(1 - \dfrac{\sf r}{\sf K}\right) .
		\end{equation}
	The obtained bound is again larger than the achievable bound in \eqref{eqn:achievable} for all $\sf K \geq 3$. 
\end{IEEEproof}

\section{The New IA \& ZF Scheme} \label{sec:scheme}

\subsection{A scheme for $\sf r = \lfloor (\sf K-1)/2\rfloor$}
Each message is cooperatively transmitted by a size-$\r$ set of nodes $\mathcal{T}$ so that it is received at a given node $k\in[\K]\backslash \mathcal{T}$ while  zero-forced at a group $\mathcal{S}$ of  $\r-1$ nodes in $[\K]\backslash\{  \mathcal{T}\cup \{k\}\}$.  There is a set $\cL \subset [\K] \backslash  \{\mathcal{S} \cup  \mathcal{T}\cup \{k\}\}$ of $\K - (\r+ \r-1+ 1) = \K - 2\r$ nodes where the signal is experienced as interference.
We choose to precode all messages that cause interference at a given set of Nodes $\mathcal{L}$ by the same precoding matrix $\mat{U}_{\cL}$, for $\mathcal{L} \in [\K]$, with $|\mathcal{L}| = \K-2\r$. 
We want to construct the  precoding matrix $\mat{U}_\ell$ so that all interferences at the same set of nodes $\cL$ align, thus leaving the remaining space for signaling dimensions. 
To summarize, if we use precoding matrix $\mat{U}_{\cL}$ for the transmission of a message from group $\mathcal{T}$ to Node $k$, then we zero-force this signal at all nodes in $[\K]\backslash \{  \mathcal{T} \cup \cL \cup\{k\}\}$. We ensure that this signal is aligned with all other interference signals precoded with $\mat{U}_{\cL}$ at the nodes $\ell \in \cL$, which are the only nodes where it causes interference.

Based on the above design principle, the key design choice in our scheme is the assignment of a precoding matrix to each message. Equivalently, for each message, we select the set of nodes $\mathcal{L}$ at which this message causes interference, and assign the corresponding precoding matrix $\mat{U}_{\mathcal{L}}$. This assignment fully determines both the interference alignment structure and the zero-forcing constraints of the scheme.

For the case $\K \geq 5$, the assignment of precoding matrices follows a cyclic construction, as summarized in Algorithm~\ref{alg:precoding_assignment}. 

\begin{algorithm}[b]
	\caption{Assignment of precoding matrices for $\sf K \geq 5$}
	\label{alg:precoding_assignment}
	\textit{Throughout Algorithm~\ref{alg:precoding_assignment}, all indices are interpreted
		according to the natural cyclic order of the set $[\K]\setminus\cL$.
		In particular, the notions of successor (e.g., $k+1$) and consecutive indices
		are defined according to this cyclic order, rather than the natural integer
		ordering.}
	\begin{algorithmic}[1]
		\Require Number of nodes $\sf K$, cooperation parameter $\r$, index set $\cL$
		
		\ForAll{$k \in [\K] \setminus \cL$}
		
		\State Fix $\r-1$ consecutive nodes
		\Statex \hspace{1em}
		$\{ k-\r+1, \ldots, k-1 \} \subseteq [\K] \setminus \cL$.
		
		\ForAll{$t \in [\K] \setminus \big( \{ k-\r+1, \ldots, k \} \cup \cL \big)$}
		\State Form the set $\cT \gets \{ k-\r+1, \ldots, k-1, t \}$
		\State Assign $\mat{U}_{\cL}$ to row $\cT$ and column $k$
		
		\If{$\sf K$ is odd \textbf{and} $k+1 \notin \cT$}
		\State Assign $\mat{U}_{\cL}$ to row $\cT$ and column $k+1$
		\EndIf
		\EndFor
		\EndFor
	\end{algorithmic}
\end{algorithm}
Notice that the size of $\cL$ is $\K - 2\r$, so we select exactly $\r$ different sets $\cT$ in column $k$ when $\sf K$ is even. We select $2\r - 1 = \sf K - 2$ different sets $\cT$ when $\sf K$ is odd.

\lo{
For example, as depicted in Table~\ref{table:K6r2}, for $\K=6$, $\r=2$, $\cL = \{5,6\}$ and $k=2$ the following sets $\mathcal{T}$ send a message to Node $2$ using precoding matrix $\mat{U}_{5,6}$: 
\begin{IEEEeqnarray}{rCl}\label{eq:T1}
	\mathcal{T} &\in& \left\{ \{1, t\}\colon t\in [6] \backslash \{1,2,5,6\} \right\} = \{ \{1, 3\}, \{1,4\}\},
\end{IEEEeqnarray}
where here we associated the index $1$ with $k-1$ and $t$ runs over the remaining set $[6] \backslash \{1,2,5,6\}$. Since $\sf K=6$ is even, no additional assignment across adjacent columns is required according to Algorithm~\ref{alg:precoding_assignment}. 
For clarity, we focus on the assignment of $\mat{U}_{5,6}$ at Node~$2$. The analysis for other precoding matrices and nodes follows analogously by symmetry.
\begin{table}[ht]
	\caption{Choice of precoding matrices in our scheme for $\sf K=6$ and $r=2$. Each signal that is precoded with matrix $\mat{U}_{5,6}$ is zero-forced at the unique Node $i \in [\sf K] \backslash (\cT \cup \{k, 5, 6\})$.}
	\centering
	\begin{tabular}{c  ||  c | c | c | c | c | c }
		$\mathcal{T} \; \backslash \; k $ &1 & 2 & 3 & 4 & 5 & 6 \\ [0.5ex] 
		\hline\hline 
		$\{1,2\}$ & x   & x   & $\mat{U}_{\{5,6\}}$ & o   & o   & o \\
		$\{1,3\}$ & x   & $\mat{U}_{\{5,6\}}$ & x   & $\mat{U}_{\{5,6\}}$ & o   & o \\
		$\{1,4\}$ & x   & $\mat{U}_{\{5,6\}}$ & o   & x   & o   & o \\
		$\{1,5\}$ & x   & o   & o   & o   & x   & o \\
		$\{1,6\}$ & x   & o   & o   & o   & o   & x \\
		$\{2,3\}$ & o   & x   & x   & $\mat{U}_{\{5,6\}}$ & o   & o \\
		$\{2,4\}$ & $\mat{U}_{\{5,6\}}$ & x   & $\mat{U}_{\{5,6\}}$ & x   & o   & o \\
		$\{2,5\}$ & o   & x   & o   & o   & x   & o \\
		$\{2,6\}$ & o   & x   & o   & o   & o   & x \\
		$\{3,4\}$ & $\mat{U}_{\{5,6\}}$ & o   & x   & x   & o   & o \\
		$\{3,5\}$ & o   & o   & x   & o   & x   & o \\
		$\{3,6\}$ & o   & o   & x   & o   & o   & x \\
		$\{4,5\}$ & o   & o   & o   & x   & x   & o \\
		$\{4,6\}$ & o   & o   & o   & x   & o   & x \\
		$\{5,6\}$ & o   & o   & o   & o   & x   & x \\
	\end{tabular}
	\label{table:K6r2}
\end{table}
}

\begin{table}[ht]
	\caption{Choice of precoding matrices in our scheme for $\sf K=5$ and $r=2$. Each signal that is precoded with matrix $\mat{U}_5$ is zero-forced at the unique Node $i\in[\sf K] \backslash (\cT \cup \{k, 5\})$.}
	\centering
	\begin{tabular}{c  ||  c  | c | c| c |c  }
		$\mathcal{T} \; \backslash \; k $ &1 & 2 & 3 & 4 & 5 \\ [0.5ex] 
		\hline\hline 
		$\{1,2\}$& x & x &  \Ufi  &  \Ufi & o \\
		$\{1,3\}$& x & \Ufi & x& \Ufi & o \\
		$\{1,4\}$& x & \Ufi & \Ufi  & x & o \\
		$\{1,5\}$& x & o & o  & o  &x  \\
		$\{2,3\}$& \Ufi &x &  x & \Ufi & o \\
		$\{2,4\}$&  \Ufi & x & \Ufi & x & o \\
		$\{2,5\}$& o&  x& o  &o &x \\
		$\{3,4\}$& \Ufi & \Ufi &x&x &o \\
		$\{3,5\}$&  o &  o &x& o  &x \\
		$\{4,5\}$& o   &  o & o &x &x  \\
	\end{tabular}
	\label{table:K5r2_first}
\end{table}
\lo{
Another example is given in Table \ref{table:K5r2_first} for $\sf K=5$, $\sf r=2$ and \Ufi. 
Consider again $k=2$. The same two transmit sets as in~\eqref{eq:T1} are obtained.
However, since $\sf K$ is odd, an additional assignment is required. Specifically, the precoding matrix $\mat{U}_{\{5\}}$ must also be assigned to the entry corresponding to $\mathcal{T}=\{1,4\}$ and $k=3$, as the transmit set $\{1,4\}$ does not include $3$.
}
\sh{
For example, as depicted in Table~\ref{table:K5r2_first}, for $\K=5$, $\r=2$, $\cL = \{5\}$ and $k=2$, the following sets $\mathcal{T}$ send a message to Node $2$ using precoding matrix $\mat{U}_{5}$: 
\begin{IEEEeqnarray}{rCl}\label{eq:T1}
	\mathcal{T} &\in& \left\{ \{1, t\}\colon t\in [5] \backslash \{1,2,5\} \right\} = \{ \{1, 3\}, \{1,4\}\},
\end{IEEEeqnarray}
where here we associated the index $1$ with $k-1$ and $t$ runs over the remaining set $[5] \backslash \{1,2,5\}$.
Since $\sf K$ is odd, an additional assignment is required according to Algorithm~\ref{alg:precoding_assignment}. Specifically, the precoding matrix $\mat{U}_{5}$ must also be assigned to the entry corresponding to $\mathcal{T}=\{1,4\}$ and $k'=k+1=3$, as the transmit set $\{1,4\}$ does not include $3$. For clarity, we focus on the assignment of $\mat{U}_{5}$ with $k=2$. The analysis for other precoding matrices and columns follows analogously by symmetry.

Another example is given in~\cite{arxiv2026} for $\sf K=6$, $\sf r=2$. 
}

\subsection{Analysis of SDoF}
If all channel coefficients after zero-forcing  that multiply the same precoding matrix are algebraically independent, an IA coding scheme can be constructed by following the framework in \cite{bi_normalized_2024} and applying the precoding matrix assignment method in Algorithm~\ref{alg:precoding_assignment}. 
\lo{For the case $\sf K=5$, the desired algebraic independence is proved analytically in the next subsection.}
\sh{For $\sf K=5$, the desired algebraic independence is proved analytically in the arXiv version~\cite{arxiv2026}.}
For larger  values of $\sf K$  however, the analysis becomes considerably more complicated and we verified algebraic independence  directly by computing the Jacobian matrix using \textsc{Matlab}\footnote{The code is at: https://github.com/yue-bi/wdc-jacobian-verification.}. This verification has been carried out for all  $\sf K \in\{6,\ldots, 15\}$. This range is chosen due to the computational complexity of the symbolic Jacobian computation.

To evaluate the SDoF achieved by our scheme, we notice that the  DoF per node is given by the ratio between the dimension of the desired signal subspace and the dimension of the desired signal subspace and  interference subspace. For sufficiently large message lengths, the ratio between the dimensions of the desired signal subspace and the interference subspace asymptotically equals the ratio between the number of codewords intended for a given Node and the number of distinct precoding matrices in the interference space.

As mentioned earlier, each column is filled with $\sf r$ matrices when $\sf K$ is even, and with $2\sf r -1 =\sf K-2$ matrices when $\sf K$ is odd. This implies that for a node~$k \notin \mathcal{L}$, the precoding matrix $\mat{U}_{\cL}$ encodes $\sf r$ messages intended to Node~$k$ when $\sf K$ is even, and it encodes $\sf K -2$ messages intended to Node~$k$ when $\sf K$ is odd. 
Since $|\cL| = \sf K - 2\sf r$, the number of precoding matrices that encode a message intended for Node~$k$ equals $\binom{\sf K-1}{\sf K - 2\sf r}.$
Meanwhile, Node $k$ is only interfered by codewords that are premultiplied by $\mat{U}_{\cL}$ with $k \in \cL$. The number of interfering precoding matrices is $\binom{\K-1}{\K - 2\r-1}$. 
Therefore, for $\sf r = \lfloor( \sf K-1)/2 \rfloor$, the following DoF is achievable to each Node $k$ when we let the length of message tends to infinity:
\begin{IEEEeqnarray}{c}
	\begin{cases} \label{eq:diff}
		\frac{\binom{\sf K-1}{\sf K - 2\sf r} \sf r}{\binom{\sf K-1}{\sf K - 2\sf r} \sf r + \binom{\K-1}{\K - 2\r-1}} = \frac{(\sf K - 2)^2}{(\sf K - 2)^2 + 4} ,& \K \text{ even,} \\[0.5ex]
		\frac{\binom{\sf K-1}{\sf K - 2\sf r} (\sf K-2)}{\binom{\sf K-1}{\sf K - 2\sf r} (\sf K -2) + \binom{\K-1}{\K - 2\r-1}} = \frac{(\sf K - 1)(\sf K - 2)}{(\sf K - 1)(\sf K - 2) + 1} ,& \K \text{ odd.}
	\end{cases}
\end{IEEEeqnarray}

The achievable SDoF can be obtained by mutiplying the above results by $\sf K$, as the scheme is symmetric.

\lo{
\subsection{Independence of channel coefficients for $\sf K=5$}
We now specify the assignment of the precoding matrices for the case $\sf K=5$ and $\sf r=2$. Table~\ref{table:K5r2_first} illustrates the assignment associated with the precoding matrix $\mat{U}_5$. We focus on the analysis of $\mat{U}_5$ in the following. The treatment of the other precoding matrices follows analogously due to the symmetry of the system.

\subsubsection{Encoding}
To highlight the role of the precoding matrix $\mat{U}_5$, we explicitly write only the components of the transmitted signals that are premultiplied by $\mat{U}_5$ in Table~\ref{table:K5r2_first}, while the remaining terms are omitted for brevity.

We encode each message $a_{k, \cT}$ into a length-$\eta^{\Gamma}$ codeword $\vect{b}_{k, \cT}$, where $\eta$ is a fixed integer and $\Gamma$ is the number of channel coefficients premultiplying $\mat{U}_5$.
Each codeword that is premultiplied by $\mat{U}_5$ causes interference at Node~$5$. The applied zero-forcing for codeword $\vect{b}_{k,\cT}$ is 
\begin{equation}
	i \in [\sf K]\setminus(\cT \cup \{k, 5\}),
\end{equation}
which is
\begin{itemize}
	\item Zero-forced at Node~1: $\vect{b}_{4,(2,3)}$, $\vect{b}_{3,(2,4)}$, $\vect{b}_{2,(3,4)}$;
	\item Zero-forced at Node~2: $\vect{b}_{4,(1,3)}$, $\vect{b}_{1,(3,4)}$, $\vect{b}_{3,(1,4)}$;
	\item Zero-forced at Node~3: $\vect{b}_{2,(1,4)}$, $\vect{b}_{1,(2,4)}$,  $\vect{b}_{4,(1,4)}$;
	\item Zero-forced at Node~4: $\vect{b}_{3,(1,2)}$, $\vect{b}_{2,(1,3)}$, $\vect{b}_{1,(2,3)}$.
\end{itemize}

Define $\cB_5$ as all index pairs $(k, \cT)$ of codewords premultiplied by $\mat{U}_5$. The $\mat{U}_5$-related components of the transmitted signal at Node~$p$ can be expressed as
\begin{IEEEeqnarray}{rCl}
	\vect{X}_p & = & \sum_{(k, \cT) \in \cB_5, p \in \cT} \mat{V}_{k, \cT}^{p} \mat{U}_5 \vect{b}_{k, \cT} + \dots, 
\end{IEEEeqnarray}
where $\mat{V}_{k, \cT}^{p} = \text{diag} \left(v_{k, \cT}^{p}(1), \dots, v_{k, \cT}^{p}(\sf n) \right)$ is the ZF matrix for codeword $\vect{b}_{k, \cT}$ at Node~$p$. 

Correspondingly, the $\mat{U}_5$-related components of the received signal at Node~$p'$  are given by
\begin{IEEEeqnarray}{rCl}
	\vect{Y}_{p'} & = & \sum_{(k, \cT) \in \cB_5, p' \notin \cT} \mat{H}_{p', \cT} \mat{V}_{k, \cT} \mat{U}_5 \vect{b}_{k, \cT} + \dots + \vect{Z}_{p'},  \IEEEeqnarraynumspace
\end{IEEEeqnarray}
where $$\mat{H}_{p', \cT} = [ \mat{H}_{p', p} ]_{p \in \cT} \quad \textnormal{and} \quad \mat{V}_{k, \cT} =  \left( [\mat{V}_{k, \cT}^{p}]_{p \in \cT} \right)^T.$$
To implement ZF, we need
\begin{equation}
	\mat{H}_{i, \cT} \mat{V}_{k, \cT} = 0, \quad \forall i \in [\sf K]\setminus(\cT \cup \{k, 5\}).
\end{equation}
Notice that both  $\mat{H}$- and $\mat{V}$-matrices are diagonal, and moreover, $|\cT| = 2$ and $|[\sf K]\setminus(\cT \cup \{k, 5\})| = 1$ for $\sf K=5$ and $\sf r=2$. 

For  $\cT = (p1, p2)$, we choose 
\begin{IEEEeqnarray}{rCl}
	v_{k, \cT}(t) &=& [v_{k, \cT}^{p1}(t), v_{k, \cT}^{p2}(t)]^T = s_{k, \cT}(t)[-h_{i, p2}(t), h_{i, p1}(t)]^T, \nonumber \\
	&& \forall t \in [\sf n], i = [\sf K]\setminus(\cT \cup \{k, 5\}),
\end{IEEEeqnarray}
where $s_{k, \cT}(t)$ is a random number independent of all others. Denote $h_{i, \cT}(t) = [h_{i, p1}(t), h_{i, p2}(t)]$ With this choice, we always have the inner product $<v_{k, \cT}(t), h_{i, \cT}(t)>=0$.

Taking into account the ZF, the  received signal at Node $p'$ precoded by $\mat{U}_5$ can be rewritten as 
\begin{IEEEeqnarray}{rCl}
	\vect{Y}_{p'} & = & \sum_{(k, \cT) \in \cB_5, k \notin \cT} \mat{G}^{(p')}_{k, \cT} \mat{U}_5 \vect{b}_{k, \cT} + \dots + \vect{V}_k, 
\end{IEEEeqnarray}
where
\begin{equation}
	\mat{G}^{(p')}_{k, \cT} \triangleq \mat{H}_{p', \cT} \mat{V}_{k, \cT}.
\end{equation}
Also notice that matrices are diagonal, i.e. 
\begin{equation}
	\mat{G}^{(p')}_{k, \cT} = \text{diag}(g^{(p')}_{k, \cT}(1), \dots g^{(p')}_{k, \cT}(\sf n)),
\end{equation}
and for $t \in [\sf n]$, the $t$-th diagonal entry is given by
\begin{equation}
	g^{(p')}_{k, \cT}(t) = s_{k, \cT}(t) \mu_{(i,p'),\cT}(t), \quad i = [\sf K]\setminus(\cT \cup \{k, 5\}),
\end{equation}
\begin{equation}
	\mu_{(i,p'),\cT}(t) \triangleq \begin{vmatrix}
		h_{i,p1}(t) & h_{i,p2}(t) \\
		h_{p',p1}(t) & h_{p',p2}(t)
	\end{vmatrix}.
\end{equation}

After ZF we are left with channel coefficients $g^{(p')}_{k, \cT}(t)$, where $p' \in \{k, 5\}$. Notice that for $p' = k$ the channel coefficient $g^{(k)}_{k, \cT}(t)$  premutiplies a useful codeword for Node~$k$. For  $p' \neq k$  the  channel coefficient $g^{(p')}_{k, \cT}(t)$ premultiplies an interfering codeword.

\subsubsection{Independence and Calculation of Jacobian matrix}
We wish to verify the algebraic independence among all channel coefficient $\mat{G}_{k,\mathcal{T}}^{(p')}$ that premultiply the matrix $\mat{U}_5$. Since the coefficients across different times $t \in [\sf n]$  are independent, we only need to analyze the independence between coefficients $g_{k,\mathcal{T}}^{(p')}(t)$ for a fixed $t$. For simplicity, we omit the index t.

We thus consider the Jacobian matrix
\begin{equation}
	\mat{J} = \left( \frac{\partial \vect{g}}{\partial (\vect{s}, \vect{h})} 	\right),
\end{equation}
where $\vect{g}$ denotes the vector of all coefficients $g_{k,\mathcal{T}}^{(p')}$, $\vect{s}$ denotes the vector of random scalars $s_{k, \mathcal{T}}$, and $\vect{h}$ denotes the original channel coefficients.

We partition the Jacobian matrix according to the variables $(\vect{s},\vect{h})$ and the nature of the coefficients as
\[
\mat{J}
=
\begin{pmatrix}
	\displaystyle \frac{\partial \vect{g}_{\mathrm{intf}}}{\partial \vect{s}}
	&
	\displaystyle \frac{\partial \vect{g}_{\mathrm{intf}}}{\partial \vect{h}}
	\\[1ex]
	\displaystyle \frac{\partial \vect{g}_{\mathrm{use}}}{\partial \vect{s}}
	&
	\displaystyle \frac{\partial \vect{g}_{\mathrm{use}}}{\partial \vect{h}}
\end{pmatrix},
\]
where $\vect{g}_{\mathrm{intf}}$ and $\vect{g}_{\mathrm{use}}$ denote the interference-related and useful coefficients, respectively.

A key observation is that the two blocks in the first column,
$\frac{\partial \vect{g}_{\mathrm{intf}}}{\partial \vect{s}}$ and
$\frac{\partial \vect{g}_{\mathrm{use}}}{\partial \vect{s}}$,
exhibit a diagonal structure.  Indeed, each scalar $s_{k,\mathcal{T}}$ appears in exactly two coefficients generated by the same message. For instance, the coefficients
$g_{3,(1,2)}^{(3)}$ and $g_{3,(1,2)}^{(5)}$ are the only terms containing the factor $s_{3,(1,2)}$, and their partial derivatives satisfy
\[
\frac{\partial g_{3,(1,2)}^{(3)}}{\partial s_{3,(1,2)}}
=
\mu_{(3,4)(1,2)},
\qquad
\frac{\partial g_{3,(1,2)}^{(5)}}{\partial s_{3,(1,2)}}
=
\mu_{(5,4)(1,2)} .
\]

Exploiting this structure, we apply Gaussian elimination to eliminate the lower-left block
$\frac{\partial \vect{g}_{\mathrm{use}}}{\partial \vect{s}}$.
Specifically, we perform the row operation
\[
\mat{J}\!\left[g_{3,(1,2)}^{(3)},:\right]
\leftarrow
\mat{J}\!\left[g_{3,(1,2)}^{(3)},:\right]
-
\frac{\mu_{(3,4)(1,2)}}{\mu_{(5,4)(1,2)}}
\mat{J}\!\left[g_{3,(1,2)}^{(5)},:\right],
\]
which sets the entry $\mat{J}[g_{3,(1,2)}^{(3)}, s_{3,(1,2)}]$ to zero. As a result, the column corresponding to $s_{3,(1,2)}$ contains a single nonzero entry, located in the row associated with $g_{3,(1,2)}^{(5)}$.

Repeating this Gaussian elimination procedure for all pairs of coefficients sharing the same subscript $(k,\mathcal{T})$, the Jacobian matrix is transformed into an upper block-triangular form
$$
\mat{J} \sim 
\begin{pmatrix}
	\displaystyle \frac{\partial \vect{g}_{\mathrm{intf}}}{\partial \vect{s}}
	&
	\displaystyle \frac{\partial \vect{g}_{\mathrm{intf}}}{\partial \vect{h}}
	\\[1ex]
	\displaystyle \mat{0}
	&
	\displaystyle \mat{J}_h
\end{pmatrix}.$$ 
Consequently, after Gaussian elimination, the rank of the full Jacobian matrix is determined by the remaining lower-right block $\mat{J}_h$.

We begin by specifying the grouping of the coefficients $g$ and the channel
coefficients $h$, which induces a natural block structure in the Jacobian
matrix.
The coefficients $g_{k,\cT}^{(k)}$ are ordered and partitioned into four groups. Specifically, the groups are defined as
\begin{align*}
	\vect{G}_1 &=
	( g_{1,(2,3)}^{(1)},\ g_{2,(3,4)}^{(2)},\ g_{3,(1,4)}^{(3)},\ g_{4,(1,2)}^{(4)} ),\\
	\vect{G}_2 &=
	( g_{1,(3,4)}^{(1)},\ g_{2,(1,4)}^{(2)},\ g_{3,(1,2)}^{(3)},\ g_{4,(2,3)}^{(4)} ),\\
	\vect{G}_3 &=
	( g_{2,(1,3)}^{(2)},\ g_{4,(1,3)}^{(4)} ),\\
	\vect{G}_4 &=
	( g_{1,(2,4)}^{(1)},\ g_{3,(2,4)}^{(3)} ).
\end{align*}
The channel coefficients are grouped into four vectors as well
\begin{align*}
	\vect{h}_1 &= (h_{12}, h_{23}, h_{34}, h_{41}),\\
	\vect{h}_2 &= (h_{13}, h_{24}, h_{31}, h_{42}),\\
	\vect{h}_3 &= (h_{21}, h_{43}),\\
	\vect{h}_4 &= (h_{32}, h_{14}).
\end{align*}
According to the grouping of coefficients $g$ and channel variables $h$, the matrix $\mat{J}_h$ is partitioned into blocks
\begin{equation}
	\mat{J}_h = \left[\mat{J}_{i,j}\right]_{i,j \in \{1,2,3,4\}},
\end{equation}
where $\mat{J}_{i,j} \triangleq \partial \vect{G}_i / \partial \vect{h}_j.$

This grouping is designed to align the dependence structure between the coefficients $g$ and the channel variables $h$, so that the Jacobian matrix admits a nearly block upper-triangular structure after suitable row and column permutations. Moreover, each diagonal block is sparse and has nonzero entries only on the main diagonal and on the immediately right-adjacent positions, with all other entries being zero. This structure enables a blockwise analysis of the Jacobian determinant via its diagonal blocks. 

We first focus on the diagonal block
\[
\mat{J}_{22} = \frac{\partial \vect{G}_2}{\partial \vect{h}_2},
\]
where
\[
\vect{G}_2
=
\{ g_{1,(3,4)}^{(1)},\ g_{2,(1,4)}^{(2)},\ g_{3,(1,2)}^{(3)},\ g_{4,(2,3)}^{(4)} \}.
\]
Each coefficient $g_{k,\cT}^{(k)}$ is a scaled $2\times2$ minor of the channel
matrix and therefore depends on exactly two channel coefficients.
Consequently, in $\mat{J}_{22}$, each column contains exactly two nonzero
entries, which appear in the same two rows.
The corresponding Jacobian block is written in \eqref{eq:J22}.
\begin{figure*}[t]
	\begin{equation} \label{eq:J22}
		\mat{J}_{22} = 
		\begin{array}{c|cccc} & h_{13} & h_{24} & h_{31} & h_{42} \\\hline g_{1,(34)}^{(1)} & h_{24} & h_{13} - \dfrac{\mu_{(12)(34)}}{\mu_{(52)(34)}}\,h_{53} & 0 & 0 \\
			g_{2,(14)}^{(2)} & 0 & h_{31} & h_{24} - \dfrac{\mu_{(23)(14)}}{\mu_{(53)(14)}}\,h_{54} & 0 \\
			g_{3,(12)}^{(3)} & 0 & 0 & h_{42} & h_{31} - \dfrac{\mu_{(34)(12)}}{\mu_{(54)(12)}}\,h_{51} \\
			g_{4,(23)}^{(4)} & \;h_{42} - \dfrac{\mu_{(14)(23)}}{\mu_{(54)(23)}}\,h_{52}\; & 0 & 0 & h_{13} 
		\end{array}
	\end{equation}
\end{figure*}
The determinant is non-zeros almost surely as the it is the product of terms on the diagonal minus the product of terms on the superdiagonal. 

After eliminating the rows and columns corresponding to $\vect{G}_2$ and $\vect{h}_2$, the remaining coefficients are
\[
g_{2,(1,3)}^{(2)},\quad g_{4,(1,3)}^{(4)},\quad
g_{1,(2,4)}^{(1)},\quad g_{3,(2,4)}^{(3)}.
\]

We next consider the diagonal block
\[
\mat{J}_{33}
\triangleq
\frac{\partial (g_{2,(1,3)}^{(2)},\, g_{4,(1,3)}^{(4)})}
{\partial (h_{21},\, h_{41})},
\]
which corresponds to the coefficient group $\vect{G}_3$ and channel group
$\vect{h}_3$.
It is given by
\begin{equation}
	\mat{J}_{33} =
	\begin{array}{c|cc}
		& h_{21} & h_{41} \\ \hline
		g_{2,(1,3)}^{(2)} &
		h_{23} &
		-\dfrac{\mu_{(5)(1,3)}}{\mu_{(5)(2,3)}}\,h_{52} \\[0.8em]
		g_{4,(1,3)}^{(4)} &
		-\dfrac{\mu_{(5)(2,3)}}{\mu_{(5)(1,3)}}\,h_{51} &
		h_{43}
	\end{array}
\end{equation}
whose determinant is nonzero almost surely. Meanwhile, we can show that the matrices $\mat{J}_{32}$ and $\mat{J}_{34}$ are both zero-matrices. This is because both channel coefficients in $\vect{G}_3$ are from Nodes $1$ and $3$, so the derivatives to a channel coefficient not from Nodes $1$ and $3$ or to Node $1$ and $3$ are zeros. 

Similar analyses apply also to the 4th line of the block matrix, which means $\mat{J}_{44}$ also has a nonzero determinant almost surely, and $\mat{J}_{42} = \mat{0}$ and $\mat{J}_{43} = \mat{0}$. Therefore, the Jacobian matrix is equivalent to
\begin{equation}
	\mat{J}_h =
	\begin{bmatrix}
		\mat{J}_{11} & \mat{J}_{12} & \mat{J}_{13} & \mat{J}_{14} \\
		\mat{J}_{21}            & \mat{J}_{22} & \mat{J}_{23} & \mat{J}_{24} \\
		\mat{J}_{31}            & 0            & \mat{J}_{33} & 0 \\
		\mat{J}_{41}            & 0            & 0            & \mat{J}_{44}
	\end{bmatrix}.
\end{equation}
Since the diagonal blocks $\mat{J}_{22}$, $\mat{J}_{33}$, and $\mat{J}_{44}$ have been shown to be nonsingular almost surely, one may apply Gaussian elimination to eliminate the blocks $\mat{J}_{21}$, $\mat{J}_{31}$, and $\mat{J}_{41}$, using $\mat{J}_{22}$, $\mat{J}_{33}$, and $\mat{J}_{44}$ as pivots, respectively. 

Although $\mat{J}_{22}$ is nonsingular almost surely, explicitly computing their inverse matrices is complex. Instead, we again invoke a polynomial argument. Since all entries of the Jacobian matrix are polynomial functions of the channel coefficients, the determinant $\det(\mat{J}')$ is itself a polynomial. To show that it is not identically zero, it suffices to exhibit one admissible channel realization for which $\det(\mat{J}_h) \neq 0$.

In particular, consider a channel realization with $\vect{s} = \vect{1}$, $[\vect{h}_{5,1}, \dots h_{5,4}] = \vect{1}$. We further impose that all superdiagonal entries of $\mat{J}_{22}$ vanish
while the diagonal entries remain nonzero. This condition is equivalent to the following set of equations:
\begin{align}
	-h_{13} - \frac{ h_{13} h_{24} - h_{14} h_{23} }{ h_{31} - h_{34} } &= 0, \label{eq:sd1} \\[0.6em]
	h_{24} + \frac{ h_{21} h_{34} - h_{24} h_{31} }{ h_{41} - h_{42} } &= 0, \label{eq:sd2} \\[0.6em]
	-h_{31} - \frac{ h_{31} h_{42} - h_{32} h_{41} }{ h_{12} - h_{13} } &= 0, \label{eq:sd3} \\[0.6em]
	\frac{ h_{12} h_{43} - h_{13} h_{42} }{ h_{23} - h_{24} } - h_{42} &= 0. \label{eq:sd4}
\end{align}
The system of equations \eqref{eq:sd1}--\eqref{eq:sd4} can be solved for the variables $h_{14}$, $h_{21}$, $h_{32}$, and $h_{43}$ as functions of the remaining channel coefficients. 
Under this choice, $\mat{J}_{22}$ becomes diagonal, and its determinant reduces to the product of its diagonal entries, which is nonzero. 
Having fixed $\mat{J}_{22}$ according to the above construction, we may further assign simple
numerical values to the remaining free channel variables. For such a choice, after performing Gaussian elimination using the diagonal blocks $\mat{J}_{22}$, $\mat{J}_{33}$, and $\mat{J}_{44}$ as pivots, it can be readily verified that the determinant of the resulting upper-left block $\mat{J}'_{11}$ is nonzero.

The matrix $\mat{J}'$ is block upper triangular, and therefore
$$ \det(\mat{J}) = \det\left(\frac{\partial \vect{g}_{\mathrm{intf}}}{\partial \vect{s}}\right) \cdot \det(\mat{J}'_{11})\cdot\prod_{i=2}^4 \det(\mat{J}_{ii}),$$
which is nonzero almost surely. This concludes the proof of the desired algebraic independence.
}

\section{Proof of the Converse Result (Theorem~\ref{thm:split_converse})} \label{sec:converse}
We recall that the converse proof is derived under the assumption that each IVA  $a_{q,p}$ is split into $\sf r$ equal-length sub-messages $(a^q_{t, \mathcal{T}})_{t\in\mathcal{T}}$ and $\mathcal{T}=\{k\colon p\in \mathcal{M}_k\}$. Denote the set of all sub-IVAs by
\begin{equation} 
\mathcal{V}_{\textnormal{total}}:= \big\{ a^q_{t, \mathcal{T}} \colon \cT \subset [[\sf K]]^{\sf r}, \ t \in \cT, \ q \in [\sf K]\backslash \cT \big\}.
\end{equation} 

\lo{
The overall proof strategy is similar to the proofs in \cite{hachem_degrees_2018} and \cite{bi_normalized_2024}. Namely, we carefully derive a converse bound on  subsystems and then combining the   obtained bounds to a bound on the entire system. However, the way we obtain the converse bound for the subsystem is more involved including a dynamic component, compared to the proofs in \cite{hachem_degrees_2018} and \cite{bi_normalized_2024}. 
}

\newcommand{\Rx}{\textnormal{Rx}}
\newcommand{\Tx}{\textnormal{Tx}}

We  have the following lemma.
\begin{lemma}\label{lma:upper_bound_lemma}
	Fix $\sf P>0$ and a rate tuple in $\mathcal{C}(\sf P)$. For any indices
	$j, t\in[\sf K]$ with $j\neq t$, the following high-SNR inequality holds for $R_a$ the rate of message $a$:
	\begin{equation}\label{eqn:thm_1_new}
		\varlimsup_{\sf P\to\infty}\;
		\frac{1}{\log \sf P} \sum_{a\in\mathcal{V}} R_a \leq \sf 1,
	\end{equation}
	where $\mathcal{V}$ is defined as 
	\begin{equation}
		\mathcal{V} \triangleq \mathcal{V}^{r} \cup \mathcal{V}^{t},
	\end{equation}
	for 	
	\begin{equation}\label{eqn:Vr_sp}
		\mathcal{V}^{\Rx} \triangleq \{ a^{j}_{u, \cT} : \cT \in \left[[\sf K]\backslash \{j\}\right]^{\sf r} \text{ and } u \in \cT \},
	\end{equation}
	and 
	\begin{equation}
		\mathcal{V}^{\Tx} \triangleq \bigcup_{i \in [\sf K]\backslash \{j,t\}} \mathcal{V}^{\Tx, i}, 
	\end{equation}
	\begin{equation}
		\mathcal{V}^{\Tx, i}\triangleq \{ a^{i}_{t, \cT} : \cT \in \left[[\sf K]\backslash \{j, \dots, i\} \right]^{\sf r} \text{ and } t \in \cT \}.
	\end{equation}
In all the above expressions, the index set $\{j, \dots, i\}$ needs to be interpreted according to the natural cyclic order of the set $[\K]\setminus\{t\}$. 
\end{lemma}

\begin{proof}
	Without loss of generality, we choose $j = 2$ and $t = 1$ throughout this proof.
	\lo{This leads to 
		\begin{equation}
		\mathcal{V}^{\Tx} \triangleq \bigcup_{i=3}^{\sf{K}}  \mathcal{V}^{\Tx, i}, 
	\end{equation}
	and for $i=3, \ldots, \sf K$:	
		\begin{equation}
		\mathcal{V}^{\Tx, i} \triangleq \left \{ a^{i}_{1, \cT} : |\cT|=\sf r, \cT \subset [\sf K]\backslash \{2, \dots, i\} \text{ and } 1 \in \cT \right\}.
	\end{equation}
	Notice that the sets $\mathcal{V}^{\Tx,  \sf K - \sf r +1}, \ldots, \mathcal{V}^{\Tx, \sf{K}}$ are all empty. For $i= 3, \ldots, \sf K - \sf r$, the message set $\mathcal{V}^{\Tx, i}$ only contains sub-IVAs intended to node $i$. Moreover, the union $\mathcal{V}^{\Tx, i} \cup \cdots \cup \mathcal{V}^{\Tx, \sf K}$ only contains sub-IVAs $a_{p, \mathcal{T}}^q$ for transmit sets $\cT$ not containing indices $2, \ldots, i$. As a consequence, the IVAs from transmit sets containing $i$ must be contained in the complement message set: 
	\begin{equation} 
	\{a_{p,\mathcal{T}}^q \colon i\in \mathcal{T}\}\subset  \big(\mathcal{V}^{\Rx} \cup \bar{\mathcal{V}} \cup \mathcal{V}^{\Tx, 3} \cup \cdots \cup \mathcal{V}^{\Tx, i-1}\big).\label{eq:inc}
	\end{equation}
}

	Fix $\sf P>0$ and any rate tuple in $\mathcal{C}(\sf P)$. Consider encoding and decoding functions $\{f_q^{(\sf n)}\}$ and $\{g^{(\sf n)}\}$ such that $p^{(\sf n)}(\mathrm{error})\to 0$ as $\sf n\to\infty$. 
	Fix a blocklength $\sf n$ and set $\mathcal{F}\triangleq \{3, \ldots, \sf{K}\}$. Let $\mathcal{H}$ denote the collection of all channel coefficients.
	For any $\mathcal{S}\subseteq[\sf K]$, define
	$\vect Y_{\mathcal{S}}\triangleq (\vect Y_p :p\in\mathcal{S})$ and
	$\vect Z_{\mathcal{S}}\triangleq (\vect Z_p:p\in\mathcal{S})$.
	All other messages are collectively denoted by
	\begin{equation}
		\bar{\mathcal{V}}\triangleq \mathcal{V}_{\textnormal{total}} \backslash {\mathcal{V}}.
	\end{equation}

	By independence of all messages, channel coefficients, and noise sequences,
	\begin{IEEEeqnarray}{rCl}\label{eq:p1_new_full}
		H(\mathcal{V})
		&=& H(\mathcal{V}\mid \bar{\mathcal{V}},\mathcal{H}) \nonumber\\
		&=& I(\mathcal{V};\*Y_{2}\mid \bar{\mathcal{V}},\mathcal{H})
		+ H(\mathcal{V}\mid \*Y_{2},\bar{\mathcal{V}},\mathcal{H}) \nonumber\\
		&=& h(\*Y_{2}\mid \bar{\mathcal{V}},\mathcal{H}) - h(\*Z_{2})
		+ H(\mathcal{V}\mid \*Y_{2},\bar{\mathcal{V}},\mathcal{H}) \nonumber\\
		&\le& \sf n\cdot \log\!\bigl(1+(\sf K -1)\sf P \sf H_{\max}^2 \bigr)
		+ H(\mathcal{V}\mid \*Y_{2},\bar{\mathcal{V}},\mathcal{H}).\IEEEeqnarraynumspace
	\end{IEEEeqnarray}
	In the following we show that the last entropy term grows slowlier than $n\log P$, which will conclude the proof of the lemma. 
	To this end, start by decomposing the entropy term as 
		\begin{IEEEeqnarray}{rCl}
	\lefteqn{H(\mathcal{V}\mid \*Y_{2},\bar{\mathcal{V}},\mathcal{H})}	\quad \nonumber \\
&=&	H(\mathcal{V}^{\Rx} \mid \*Y_{2},\bar{\mathcal{V}},\mathcal{H}) +	H(\mathcal{V}^{\Tx}\mid \*Y_{2},\bar{\mathcal{V}},\mathcal{V}^{\Rx} ,\mathcal{H}).  \IEEEeqnarraynumspace
		\end{IEEEeqnarray}	
		The first term grows sublinearly in $n$ by Fano's inequality because  
	 the messages in $\mathcal{V}^{\Rx}$ need to  be decoded
	from $\*Y_2$:
	\begin{equation}\label{eq:fano_Vr_full}
		H(\mathcal{V}^{\Rx} \mid \*Y_2,\bar{\mathcal{V}},\mathcal{H})
		\le \sf n\,\epsilon_{\sf n},
	\end{equation}
	for some $\epsilon_{\sf n}\to 0$ as $\sf n\to\infty$.

For the second term, we write
		\begin{IEEEeqnarray}{rCl}
H(\mathcal{V}^{\Tx} \mid \*Y_{2},\bar{\mathcal{V}},\mathcal{V}^{\Rx} ,\mathcal{H})
	&= & H(\mathcal{V}^{\Tx} \mid \*Y_{2},  \*Y_{\mathcal{F}},\bar{\mathcal{V}},\mathcal{V}^{\Rx} ,\mathcal{H})\nonumber \\
	&& + 	I(\mathcal{V}^{\Tx};\*Y_{\mathcal{F}} \mid
		\mathcal{V}^{\Rx},\*Y_2,\bar{\mathcal{V}},\mathcal{H}), \IEEEeqnarraynumspace
		\end{IEEEeqnarray}	
	and observe that 
	\begin{IEEEeqnarray}{rCl}
		\lefteqn{ 
		I(\mathcal{V}^{\Tx};\*Y_{\mathcal{F}} \mid
		\mathcal{V}^r{\Rx}\*Y_2,\bar{\mathcal{V}},\mathcal{H})\qquad 
	 } \nonumber\\
		&\stackrel{(a)}{=}&
		h(\tilde{\*Y}_{\mathcal{F}} \mid
		\mathcal{V}^{\Rx},\tilde{\*Y}_2,\bar{\mathcal{V}},\mathcal{H})
		- h(\*Z_{\mathcal{F}}) \\
		&\stackrel{(b)}{=}&
		h(\tilde{\*Y}_{\mathcal{F}} \mid
		\mathcal{V}^{\Rx},\tilde{\*Y}_2,\bar{\mathcal{V}},\mathcal{H},E=1)
		- h(\*Z_{\mathcal{F}})\\
		&\stackrel{(c)}{\le}&
		h(\*Z_{\mathcal{F}}' \mid
		\mathcal{V}^{\Rx},\tilde{\*Y}_2,\bar{\mathcal{V}},\mathcal{H},E=1)
		- h(\*Z_{\mathcal{F}})\\
		&\le&
		h(\*Z_{\mathcal{F}}') - h(\*Z_{\mathcal{F}})
		+ \sf n\tilde{\epsilon}_{\sf n}, \label{eq:Vt_chain_ae}
	\end{IEEEeqnarray}
	where $\tilde{\epsilon}_{\sf n}$ is a sequence that vanishes in $\sf n$ and  the (in)equalities are justified as follows:
	\begin{itemize}
		\item In $(a)$, we define 
		$
		\tilde{\vect Y}_p \triangleq \mat H_{p,t} \vect X_t + \vect Z_p, \quad p \in [\sf K].
		$
		and $\tilde{\*Y}_{\mathcal{S}} \triangleq (\tilde{\vect Y}_p : p\in\mathcal{S})$. Since $\mathcal{V}^{\Tx}$ is carried only by the input of Node~$t$, conditioning on $(\mathcal{V}^{\Rx},\bar{\mathcal{V}},\mathcal{H})$ removes all other signal components, yielding the equality in $(a)$.
		
		\item In $(b)$, we introduce the binary random variable $E$, which equals $1$ if for all $t'\in[\sf n]$, the channel coefficient $ H_{j,t}(t') \neq 0$. Notice that $\Pr[E=1]=1$ since the channel coefficients are drawn independently from  continuous distributions.
		
		\item To see $(c)$, define for each  $p\in\mathcal{F}$ the noise-modified variables 
		$ \vect Z_p' \triangleq \vect Z_p - \vect H_{p,t} \vect H_{j,t}^{-1} \*Z_j,$
		and $\*Z_{\mathcal{F}}'\triangleq (\vect Z_p' \colon p\in\mathcal{F})$. Notice that in the event $E=1$, the signal component of $\tilde{\*Y}_{\mathcal{F}}$ can be perfectly canceled, leaving only the effective noise $\*Z_{\mathcal{F}}'$. 
		
	\end{itemize}

	We finally show 
		\begin{IEEEeqnarray}{rCl}
		\lefteqn{H\left(\mathcal{V}^{\Tx} \mid \vect{Y}_2, \vect{Y}_{\mathcal{F}}, \mathcal{V}^{\Rx}, \bar{\mathcal{V}}, \mathcal{H} \right)}     \nonumber \\
		&=& \sum_{i=3}^{\sf K}
		H\left( \mathcal{V}^{{\Tx},i} \mid \vect{Y}_2, \vect{Y}_{\mathcal{F}}, \mathcal{V}^{\Rx}, \bar{\mathcal{V}}, \mathcal{V}^{{\Tx},3},\dots, \mathcal{V}^{{\Tx},i-1}, \mathcal{H}
		\right) \nonumber \\\\
		&\leq& \sum_{i=3}^{\sf K} \epsilon_{i,\sf n} \sf n , \label{eq:fano_Vt_full}
	\end{IEEEeqnarray}
	where $\epsilon_{i, \sf n}$ denotes a vanishing sequence for each $i$ and  the last inequalty holds by the following considerations. 
	
	For example, notice that the sub-IVAs known to node 3, i.e., $\{a_{p,\cT}^q\colon 3 \in \cT\}$ are all contained in the set  $(\mathcal{V}^{\Rx}\cup \bar{\mathcal{V}})$. Moreover, $\mathcal{V}^{\Tx,3}$ only contains sub-IVAs intended for node 3.  Since $ \vect{Y}_{\mathcal{F}}$ contains $ \vect{Y}_{3}$, we thus have by Fano's Inequality equality: 
	\begin{equation}
		H\left(\mathcal{V}^{\Tx,3} \mid \vect{Y}_2, \vect{Y}_{\mathcal{F}}, \mathcal{V}^{\Rx}, \bar{\mathcal{V}} \right)
		\le \epsilon_{3, \sf n} \sf n,
	\end{equation}
	for a sequence $\epsilon_{3, \sf n}\to 0$ as $\sf n\to\infty$.
	
	\lo{As noticed in \eqref{eq:inc}, for each $i=4, \ldots, \sf K$,}
	\sh{For each $i=4, \ldots, \sf K$,} the sub-IVAs known to Node $i$, i.e., $\{a_{p,\cT}^q\colon i \in \cT\}$,  are contained in the  set $(\mathcal{V}^{\Rx}\cup \bar{\mathcal{V}} \cup \mathcal{V}^{\Tx,3} \cup \cdots\cup\mathcal{V}^{\Tx,i-1})$. Since $\mathcal{V}^{\Tx,i}$ only contains sub-IVAs intended for node $i$ and $\vect{Y}_i$ is contained in $\vect{Y}_{\mathcal{F}}$, we have  by Fano's inequality: 
	\begin{equation}
		H\left(\mathcal{V}^{\Tx,i} \mid \vect{Y}_2, \vect{Y}_{\mathcal{F}} \mathcal{V}^r, \bar{\mathcal{V}}, \mathcal{V}^{\Tx,3} , \ldots ,\mathcal{V}^{\Tx,i-1} \right)
		\le \epsilon_{i, \sf n} \sf n,
	\end{equation}
	for a sequence  $\epsilon_{i, \sf n}\to 0$ as $\sf n\to\infty$. This proves \eqref{eq:fano_Vt_full}. 		
	Combining \eqref{eq:p1_new_full}--\eqref{eq:fano_Vt_full}, we obtain
	\begin{IEEEeqnarray}{rCl}
		H(\mathcal{V})
		&\le&
		\sf n\cdot \log\bigl(1+ (\sf K -1)\sf P \sf H_{\max}^2 \bigr)
		+ \sf n\,\epsilon_{\sf n} \nonumber \\
		&&+ h(\*Z_{\mathcal{F}}')-h(\*Z_{\mathcal{F}})
		+ \sf n\,\tilde{\epsilon}_{\sf n}.
	\end{IEEEeqnarray}
	Dividing by $\sf n\log \sf P$ and letting $\sf n\to\infty$ and $\sf P\to\infty$
	establishes \eqref{eqn:thm_1_new}.
\end{proof}

The proof is obtained by summing the upper bound in Lemma~\ref{lma:upper_bound_lemma} over all  index pairs $(t, j) \in [\sf K]$ with $t \neq j$,  and then normalizing. 
\lo{Each rate term appears the same number of times, namely 
\begin{equation} 
\frac{\sf K (\sf K-1) |\mathcal{V}|}{|\mathcal{V}_{\textnormal{total}}|},
\end{equation} 
and since we have $\sf K (\sf K -1)$ bounds, we obtain:
\begin{equation}\label{eq:Su}
	\mathbf{SDoF} \leq  \frac{|\mathcal{V}_{\textnormal{total}}|}{|\mathcal{V}|}.
\end{equation}
We have 
\begin{equation} \label{eq:Vtot}
|\mathcal{V}_{\textnormal{total}}| = \sf r \binom{ \sf K }{ \sf r } (\sf K -\sf r).
\end{equation} 
To compute the size of $\mathcal{V}$, notice that for each $i\leq  \sf K - \sf r$, the size of $\mathcal{V}^{\Tx,\,i}$ is
\begin{equation}
	|\mathcal{V}^{\Tx,i}|
	= \binom{\sf K-i}{\sf r-1},
\end{equation}
because to complete the transmit set $\cT$,  $\sf r -1$  additional indices need to be chosen from the $\sf K-i$ possible elements
\[
[\sf K]\setminus\{1,2,\dots,i\} = \{i+1,\dots,K\}.
\]
Using the Hockey-stick identity, we obtain
\begin{equation}
	|\mathcal{V}^{\Tx}|
	= \sum_{i=3}^{\sf K-\sf r+1} \binom{\sf K-i}{\sf r-1} = \binom{\sf K-2}{\sf r}.
\end{equation}
Moreover, 
\begin{equation}
|\mathcal{V}^{\Rx}| = \binom{\sf K-1}{\sf r} \cdot \sf r, 
\end{equation} and 
thus 
\begin{equation} 
|\mathcal{V}| = \binom{\sf K-2}{\sf r} +  \binom{\sf K-1}{\sf r} \cdot \sf r.\label{eq:Vl}
\end{equation} 
Combining \eqref{eq:Su} with \eqref{eq:Vtot} and \eqref{eq:Vl}, we obtain
\begin{equation}
	\mathbf{SDoF} \leq \frac{\sf r \binom{ \sf K }{ \sf r } (\sf K -\sf r) }{ \binom{\sf K-2}{\sf r} +  \binom{\sf K-1}{\sf r} \cdot \sf r }
	= \frac{\sf K (\sf K - 1) \sf r}{(\sf K - 2)\sf r + \sf K-1}.
\end{equation}
}
\sh{
We obtain $ \mathbf{SDoF} \leq |\mathcal{V}_{\textnormal{total}}|/|\mathcal{V}|$, thus
\begin{equation}
	\mathbf{SDoF} \leq \frac{\sf r \binom{ \sf K }{ \sf r } (\sf K -\sf r) }{ \binom{\sf K-2}{\sf r} +  \binom{\sf K-1}{\sf r} \cdot \sf r }
	= \frac{\sf K (\sf K - 1) \sf r}{(\sf K - 2)\sf r + \sf K-1}.
\end{equation}
}

\section{Conclusion}
\lo{
We proposed a new wireless MapReduce coding scheme based on interference alignement (IA) and zero-forcing (ZF) for computation loads   $\sf r =\lfloor (\sf K -1)/2\rfloor$ and show that it improves over all previous achievable NDTs. We further provide a converse result (lower bound on the NDT) for all non-cooperative coding schemes where IVAs are split into submessages and these submessages are transmitted (in an arbitrary fashion) by a single node only, thus preventing cooperation between nodes. This converse result lies above the NDT of our new IA \& ZF scheme, thus proving strict suboptimality of non-cooperative transmissions for wireless MapReduce.}
\sh{
We propose a new wireless MapReduce coding scheme based on IA and ZF for computation load $\sf r=\lfloor(\sf K-1)/2\rfloor$, which strictly improves upon all previous achievable NDTs. We also establish a converse for non-cooperative schemes, where IVAs are split into submessages and transmitted by single nodes. The results show that such schemes are strictly suboptimal compared to our IA–ZF approach.
}

 \section*{Acknowledgements}
The work of the authors has been supported by the ERC under grant agreement 101125691.

\bibliographystyle{IEEEtran}
\bibliography{IEEEabrv,references}

\end{document}